\documentclass[12pt]{amsart}

\usepackage{amsthm,amsfonts,amsmath,amsthm, amssymb}
\usepackage{verbatim}
\usepackage{graphics,graphicx,colortbl}
\numberwithin{equation}{section}
\definecolor{Mygrey}{gray}{0.8}
\newcommand{\bea}{\begin{eqnarray}}
\newcommand{\eea}{\end{eqnarray}}
\newcommand{\be}{\begin{eqnarray*}}
\newcommand{\ee}{\end{eqnarray*}}
\newtheorem{theorem}{Theorem}[section]

\newtheorem{algorithm}{Algorithm}[section]
\newtheorem{remark}{Remark}[section]

\begin{document}
\title[Detecting Fixed Points of Boolean Networks]{An Algorithm for Detecting Fixed Points of Boolean Networks}
\author[Yi Ming Zou]{Yi Ming Zou}
\address{Department of Mathematical Sciences, University of Wisconsin-Milwaukee, Milwaukee, WI 53201, USA} \email{ymzou@uwm.edu}
\begin{abstract}
In the applications of Boolean networks to modeling biological systems, an important computational problem is the detection of the fixed points of these networks. This is an NP-complete problem in general. There have been various
attempts to develop algorithms to address the computation need for large size Boolean networks. The existing methods are usually based on known algorithms and thus limited to the situations where these known algorithms can apply. In this paper, we propose a novel approach to this problem. We show that any system of Boolean equations is equivalent to one Boolean equation, and thus it is possible to divide the polynomial equation system which defines the fixed points of a Boolean network into subsystems that can be solved easily. After solving these subsystems and thus reducing the number of states, we can combine the solutions to obtain all fixed points of the given network. This approach does not depend on other algorithms and it is straightforward and easy to implement. We show that our method can handle large size Boolean networks, and demonstrate its effectiveness by using MAPLE to compute the fixed points of Boolean networks with hundreds of nodes and thousands of interactions.  
\end{abstract}
\maketitle
\date{}
\section{Introduction}
\par
\par
Boolean networks were introduced in \cite{Kauf69} as random models of genetic regulatory networks to study biological systems. A recent research focus of Boolean networks is to develop theories and algorithms to address questions arise from biological applications \cite{Cho06}, \cite{Col04}, \cite{Dev03}, \cite{Kauf04}-\cite{Pal05}, \cite{Tamu09}, \cite{Zou10}.  To aid the study of complex biological systems, where experiments are usually expensive and time consuming, researchers use mathematical models built based on partial experimental information of these biological systems. Boolean networks offer relatively simpler such models which are capable of capturing some of the key dynamical properties \cite{Bonn08},\cite{Kar08}, such as the stable states, of the underlying systems. As discrete time finite state dynamical systems, Boolean networks will eventually revert to certain sets of states called attractors. These attractors encode the long term behaviors of the underlying biological systems, and can be divided into two categories: stable states (fixed points) and cyclic states. The purpose of this paper is to develop an effective method for detecting the fixed points of these networks. 
\par
Different approaches for the detection of fixed points of Boolean networks exist in the literature. In \cite{Zhan07}, an approach which is search/recursive in nature was given. According to \cite{Zhan07}, the proposed algorithms can identify all fixed points of a random Boolean network with maximum indegree $2$ (the number of variables that each nodes depends on) with an average time $O(1.19^n)$ ($n$ is the number of variables the whole Boolean network depends on, which is also the number of nodes). In the worst case, however, it can take up to time $nO(2^n)$. Another approach, which is based on the $k$-satisfiability problem related algorithms and methods, has been developed in several recent publications (see \cite{Dev03}, \cite{Tamu09}, and the references therein). In \cite{Dev03}, the result of applying algorithms of solving constraint satisfaction problems to the detection of fixed points of some randomly generated Boolean networks was reported. According to \cite{Dev03}, this method performs well for Boolean networks with indegree $\le 2$, and the computation will be exponential with indegree $> 2$. This is because there exist polynomial time algorithms for the $k = 2$ satisfiability problem, but the satisfiability problem is NP-complete for $k>2$ \cite{Zhan07}, \cite{Tamu09}. According to \cite{Tamu09}, the algorithm there can detect a fixed point of an AND/OR (only one of these operations is allowed for each node) Boolean network with non-restricted indegree in time $O(1.787^n)$. The satisfiability problem concerns whether or not there is a solution, not how to find all solutions. Thus further developments are needed before these algorithms can find more real applications. A computational algebra approach to the theory of dynamical systems over finite fields was developed in \cite{Col04}, \cite{Jar07}, \cite{Laub04}. A key concept of computational algebra, the Gr\"{o}bner bases of a polynomial system, can also be employed in the detection of the fixed points of a Boolean network by first compute a Gr\"{o}bner basis of the polynomial system that defines the fixed points, and thus make the system easier to solve. Though the computation of a Gr\"{o}bner basis of a polynomial system over a finite field is faster than the computations over the real or the complex numbers, especially for the Boolean case, our capability to perform such a computation is still rather limited, due to the fact that the computation of a Gr\"{o}bner basis may require time that is doubly-exponential. 
\par
There are also publications concerning the connection between fixed points and the topology structures of Boolean networks. The fact that genetic networks with canalyzing Boolean rules are always stable was reported in \cite{Kauf04}, and these Boolean networks were subsequently studied by \cite{Jar2}. The problem of when a Boolean network in which all the up-dating rules are defined by monomials is a fixed point system was investigated in \cite{Col04}. In \cite{Cho06}, via minimizing a cost function over a family of Boolean networks having a common set of fixed points, the intervention in a family of Boolean networks was studied. In \cite{Xiao07}, the impact of function perturbations to a Boolean network's fixed points in the form of a one-bit change of the truth table was investigated. In \cite{Zou10}, the consistency of partial information on a Boolean network and a given set of fixed points was considered, and a testable necessary and sufficient condition for consistency was derived. Some discussions on the effects of topology of Boolean networks to their long term behaviors can be found in \cite{NL06}, \cite{Oik06}, \cite{Pome09}.
\par
The main result of this paper is a method for solving systems of Boolean equations arise from applications to biological systems. It is known that, though there are hardly any biological networks (for example, gene regulatory networks) with each and every node depends on $\le 2$ other nodes, these networks are also not densely connected. There can be nodes with many connections, but most nodes depend on a few other nodes. We have observed that, though it is impossible to treat large size Boolean network using the exhaustive enumeration method since a Boolean network with $n$ nodes has $2^n$ states, with today's standard home and office PCs, the computing time for solving a system of Boolean equations for $n<25$ variables, even using exhaustive enumeration, is rather short. For example, if $n = 20$, then the computation usually takes about $10$ seconds (see examples later). Therefore, if we can divide a system of Boolean equations into subsystems {\it according to the number of variables involved} (different subsystems can have common variables) such that each subsystem can be solved easily, say by using the exhaustive enumeration method, then by patching the solutions of the subsystems together, we should be able to find all the solutions. This turned out working quite well, since biological Boolean networks usually permit such a division, and the number of solutions of a subsystem with $t$ variables is $\ll 2^t$. So by {\it solving the subsystems first, we can reduce a seemingly intractable enumeration problem to a feasible one}.  This approach is straightforward, does not rely on any other algorithm, and is capable of solving large systems.  This method also applies to Boolean networks which have not been considered in the literature so far (to the best of the author's knowledge). For example, it applies to ``community-like'' networks, i.e. those networks where the nodes in each community (with reasonable size) can be densely connected while the communities of the network are sparsely connected, these networks can have the average number of connections of each node $> 2$.  
\par

\section{Theory and Algorithm}
\par
A Boolean network with $n$ nodes can be given by a Boolean polynomial function
\bea\label{e1}
\mathbf{f}=(f_1,\ldots,f_n):\{0,1\}^n\rightarrow\{0,1\}^n,
\eea
where $\{0,1\}^n$ is the state space of all sequences of length $n$ formed by $0$ and $1$, and $f_1,\ldots,f_n$ are Boolean polynomials in $n$ variables $x_1,\ldots, x_n$. We can use either the logical operations {\bf OR} ($\vee$), {\bf AND} ($\wedge$), and {\bf NOT} ($\neg$), or the modulo $2$ arithmetic operations addition and multiplication, to perform the calculations for Boolean variables and polynomials. The correspondences are given by:
\be
x_i\wedge x_j = x_ix_j,\;\;
x_i\vee x_j = x_i+x_j+x_ix_j,\;\; \neg x_i = x_i+1.
\ee
\par
To study the dynamical properties of a Boolean network, we consider the time-discrete dynamical system defined by:
\be
\mathbf{f}:(x_1(t),\ldots,x_n(t))\mapsto (x_1(t+1),\ldots,x_n(t+1)).
\ee
That is, the functions $f_i,\;1\le i\le n$, give the updating rules for the nodes, and the state of the $i$th node at time $t+1$ is given by the function value $f_i(x_1(t),\ldots,x_n(t))$. For gene regulatory networks, the variables $x_1,\ldots,x_n$ represent the genes and the functions $f_1,\ldots,f_n$ give the gene regulatory rules. If $x_i=1$, then the corresponding gene is expressed ({\bf ON}); and if $x_i=0$, then the gene is not expressed ({\bf OFF}). 
\par
The {\it state space graph} of a Boolean network $\mathbf{f}$ is a directed graph with the vertices (states) given by the set $\{0,1\}^n$, and with the directed edges defined by the function $\mathbf{f}$: there is a directed edge from vertex $\mathbf{v}_1$ to vertex $\mathbf{v}_2$ if the value of $\mathbf{f}$ at $\mathbf{v}_1$ is $\mathbf{v}_2$. The {\it dependency graph} of $\mathbf{f}$ is a directed graph with $n$ nodes such that there exists a direct edge from node $i$ to node $j$ ($i=j$ is allowed) if and only if $f_j$ depends on the variable $x_i$. 
\par
A state $\mathbf{x} = (x_1,x_2,\ldots,x_n)\in \{0,1\}^n$ is a fixed point of $\mathbf{f}$ if it is a solution of the system of equations 
\bea\label{e2}
f_i(x_1,x_2,\ldots,x_n) = x_i,\quad 1\le i\le n.
\eea
To describe our method, we change the above system of equations to a different equivalent form. We consider the set of Boolean polynomials
\bea\label{e3}
g_i := f(x_1,x_2,\ldots,x_n)+ x_i+1, \quad 1\le i\le n,
\eea
and let
\bea\label{e4}
m_{\mathbf{f}} = \prod_{i=1}^{n}g_i.
\eea
Let $[1,n]:=\{1,2,\ldots,n\}$. If $A\subseteq [1,n]$, we write
\bea\label{e5}
m_A = \prod_{i\in A}g_i.
\eea
Recall that a set $\{A_j\;|\;1\le j\le k\}$ of nonempty subsets of $[1,n]$ is a {\it partition} of $[1,n]$ if 
\be
\bigcup_{j=1}^kA_j = [1,n]\quad\mbox{and}\quad A_s\cap A_t = \emptyset,\;\forall s\ne t.
\ee
\par
We can now state the following theorem.
\begin{theorem} Let $\mathbf{f}$ be defined by (\ref{e1}) and let $\{A_j\;|\;1\le j\le k\}$ be a partition of $[1,n]$. Then a state $\mathbf{a} = (a_1,a_2,\ldots,a_n)\in \{0,1\}^n$ is a fixed point of $\mathbf{f}$ if and only if
\bea
m_{A_j}(a_1,a_2,\ldots,a_n) = 1,\quad\forall\; 1\le j\le k. 
\eea
In particular, $\mathbf{a}$ is a fixed point of $\mathbf{f}$ if and only if $m_{\mathbf{f}}(\mathbf{a})=1$.
\end{theorem}
\begin{proof} The system of equations given by (\ref{e2}) is equivalent to
\be
\bigvee_{i=1}^n(f_i+x_i) = 0,
\ee
which in turn is equivalent to
\be
1 = \neg(\bigvee_{i=1}^n(f_i+x_i)) = \bigwedge_{i=1}^n(f_i + x_i + 1)= \prod_{i=1}^ng_i=m_{\mathbf{f}}.
\ee
If $\{A_j\;|\;1\le j\le k\}$ is a partition of $[1,n]$, then 
\be
m_{\mathbf{f}} = \prod_{j=1}^km_{A_j},
\ee
so $m_{\mathbf{f}}=1$ if and only if all $m_{A_j}=1$.
\end{proof}
\par
\begin{remark}
By using the above argument, one can convert a satisfiability problem to a problem of finding a fixed point of a Boolean network, and vice versa. This implies immediately that detecting a fixed point of a Boolean network is an NP-complete problem.
\end{remark}
\begin{remark}
The above theorem also implies that any system of Boolean equations is equivalent to a single Boolean equation. 
\end{remark}
As a consequence of the above theorem, we have the following procedure of detecting the fixed points of a Boolean network (i.e. an algorithm for solving a system of Boolean equations).
\par
\begin{algorithm}\label{a1} Boolean network fixed points detection algorithm.
\par
\par\medskip
{\bf INPUT:} A Boolean network $\mathbf{f}=(f_1,f_2,\ldots,f_n)$ defined as in (\ref{e1}).
\par
{\bf OUTPUT:} Fixed points of $\mathbf{f}$.
\par
\medskip
1. Choose a threshold level $T$ (a positive integer) such that any Boolean equation with the number of variables $\le T$ can be solved easily.  
\par
2. Set $g_i = f_i+x_i+1,\;1\le i\le n$. Simplify the system (reduce the number of variables) using obvious relations such as $f_i = x_j$ or $f_i = x_j+1$ (for $i\ne j$) by making the substitutions $x_i=x_j$ or $x_i=x_j+1$ into the $g_i$'s. 
\par
3. Divide $[1,n]$ into subsets $\{A_j\;|\;1\le j\le k\}$ such that for each $1\le j\le k$, the number of variables involved in the subsystem $\{g_i\;:\; i\in A_j\}$ is $\le T$ (but as close to $T$ as possible), and solve each subsystem separately. 
\par
4. Combine the solutions of each subsystem to obtain the fixed points of $\mathbf{f}$.  
\end{algorithm}
\par
\begin{remark} Note that the threshold level $T$ depends on the hardware and the method employed to solve these equations. For exhaustive enumeration method on standard PCs, we can use $T=21$. Note that different subsystems are allowed to have common variables, and for each $1\le j\le k$, one can just solve $m_{A_j}=1$ (or $m_{A_j}+1=0$). Note also that parallel computation can be used in both step 3 and step 4.
\end{remark}
\par
{\it Algorithm analysis}. It is clear that the success of the above algorithm depends on whether the whole system can be divided into subsystems according to the threshold level such that the number of subsystems ($k$) is relatively small compare to the total number of nodes ($n$). For example, this will not be the case if the dependency graph is a complete graph. As mentioned in the introduction, biological networks as well as community-like networks can be divided. Basically, Boolean networks with small average connections, for example $\le 5$, can always be divided, but those with average connections $> 5$ may or may not be divisible depending on the actual networks and the method employed to solve them. Assume that exhaustive enumeration is used to solve the subsystems. From the actual gene regulatory networks in the literature, we can assume that, with the threshold level $T = 21$, the average number of equations in each subsystem is between $20$ and $30$ (see examples in the next section). If solving one of these subsystems takes about $10$ seconds, then the total time of solving these subsystems is approximately equal to $n/2$ seconds. So the computation time is up to the time needed for combining the solutions of these subsystems. This depends on the number of fixed points of the Boolean network $\mathbf{f}$. In general, the more fixed points $\mathbf{f}$ has, the longer the computation (compare the examples in the next section), since if $\mathbf{f}$ has a large number of fixed points, then even verifying that all these points are fixed points could be a problem. 
\par
\section{Examples}
\par
In this section, we present several examples for our algorithm. The first three are gene regulatory networks from the references. The last two were simulated based on the gene regulatory networks published in the literature. The subsystems were solved using exhaustive enumerations. All computations were done using MAPLE $11$ on a Dell laptop with the system: Intel(R)Core(TM)2 Duo CPU T9900@3.06GHz with 3.5 GB RAM.
\par\medskip
{\bf Example 3.1.} Our first example is the gene regulatory network published in \cite{Albe03}. This Boolean network models the expression pattern of the segment polarity genes in the fruit fly {\it Drosophila melanogaster} and has $21$ nodes. The polynomial system is given in the Appendix. There is no need to divide the system, after step 2 of Algorithm \ref{a1}, the resulted equation to be solve is
{\tiny
\be
{} &{}& 1=((x_{15}+1)*(x_{1}*(x_{2}+x_{14})+x_{2}*x_{14})+x_{2}+1)\\
{} &{}&*(x_{1}*(x_{16}*(x_{17}+1)+x_{17})+x_{16}*(x_{17}+1)+x_{17}+x_{4}+1)\\
{} &{}&*(x_{4}*(x_{15}+1)+x_{6}+1)\\
{} &{}&*((x_{4}+1)*((x_{11}+1)*(x_{20}*(x_{21}+1)+x_{21})+x_{11})+x_{8}+1)\\
{} &{}&*((x_{8}+1)*x_{9}*(x_{18}+1)*(x_{19}+1)+x_{8}+x_{9}+1)\\
{} &{}&*(((x_{8}+1)*x_{9}*(x_{18}+1)*(x_{19}+1)+x_{8})*(x_{20}*(x_{21}+1)+x_{21})+x_{10}+1)\\
{} &{}&*((x_{8}+1)*x_{9}*(x_{18}+1)*(x_{19}+1)+x_{8}+((x_{8}+1)*x_{9}\\
{} &{}&*(x_{18}+1)*(x_{19}+1)+x_{8})*(x_{20}*(x_{21}+1)+x_{21})+x_{11}+1)\\
{} &{}&*((x_{4}+1)*((x_{11}+1)*(x_{21}+1)*(x_{20}+1)+1)+x_{14}+1)\\
{} &{}&*((x_{4}+1)*((x_{11}+1)*(x_{21}+1)*(x_{20}+1)+1)+x_{4}+x_{15}).
\ee}
The computation for solving this equation took $0.54$ second, and $176$ fixed points were detected (see supplement MAPLE worksheet).
\par
\medskip 
{\bf Example 3.2.} This example is the T-LGL survival signaling Boolean network given by the diagram of Fig. 2B in \cite{ZhangR08}. This network has $29$ nodes (see Appendix). The equation we obtained after step 2 of Algorithm \ref{a1} is 
{\tiny
\be
{}&(x_{1}+x_{7}+1)*(x_{1}+x_{9}+x_{1}*x_{9}+x_{8}+1)*(x_{1}+x_{9}+x_{1}*x_{9}+x_{12})\\
{}&*(x_{15}*x_{1}*x_{9}+x_{15}*x_{1}+x_{15}*x_{9}+x_{18})*(x_{18}*(x_{1}+1)+x_{20}+1)\\
{}&*(x_{1}+x_{9}+x_{1}*x_{9}+x_{13})*(x_{9}+x_{15}*x_{9}+1)*(x_{9}+x_{28}+1) = 1.
\ee}
The computation for solving this equation took $0.45$ second, and $6$ fixed points were detected (see supplement MAPLE worksheet).
\par
\medskip
{\bf Example 3.3.} This example is the T-Cell receptor signaling Boolean model given by Fig. 1 in \cite{SR07}. This network has $90$ nodes. We derived the Boolean polynomial functions of the nodes according to the interactions given in the diagram (see Appendix). After step 2 of Algorithm \ref{a1}, with the number of variables threshold at $< 21$, the entire system of equations that defines the fixed points was divided into $3$ subsystems. Subsystem 1 involves $20$ variables and $28$ equations, which was solved in $10.32$ seconds. Subsystem 2 involves $20$ variables and $42$ equations, which was solved in $9.99$ seconds. Subsystem 3 involves $13$ variables and $20$ equations, which was solved in $0.40$ second. Putting the solutions of these subsystems together took $2.13$ seconds. A total of $4096$ fixed points were detected (see supplement MAPLE printout).
\par
\medskip
{\bf Example 3.4.} This is a simulated Boolean network with $228$ nodes. The dependency graph (Fig. 1) is generated from the polynomial functions (see supplement MAPLE printout). With threshold at $\le 21$, after step 2 of Algorithm \ref{a1}, the fixed point system of equations was divided into $7$ subsystems. The solutions of subsystems 1 to 4 were combined first, then the solutions of subsystems 5 to 7 were combined, and finally the resulted two sets of solutions were combined to obtain the fixed points. The total computation time was approximately $4.8736$ hours and a total of $25165824$ fixed points were detected. The majority of the computation time was used to combine the solutions of subsystems 1 to 4 with the solutions of subsystems 5 to 7.
\par
\begin{figure}[h]
\begin{center}
\includegraphics[width=3in, height=2.8in]{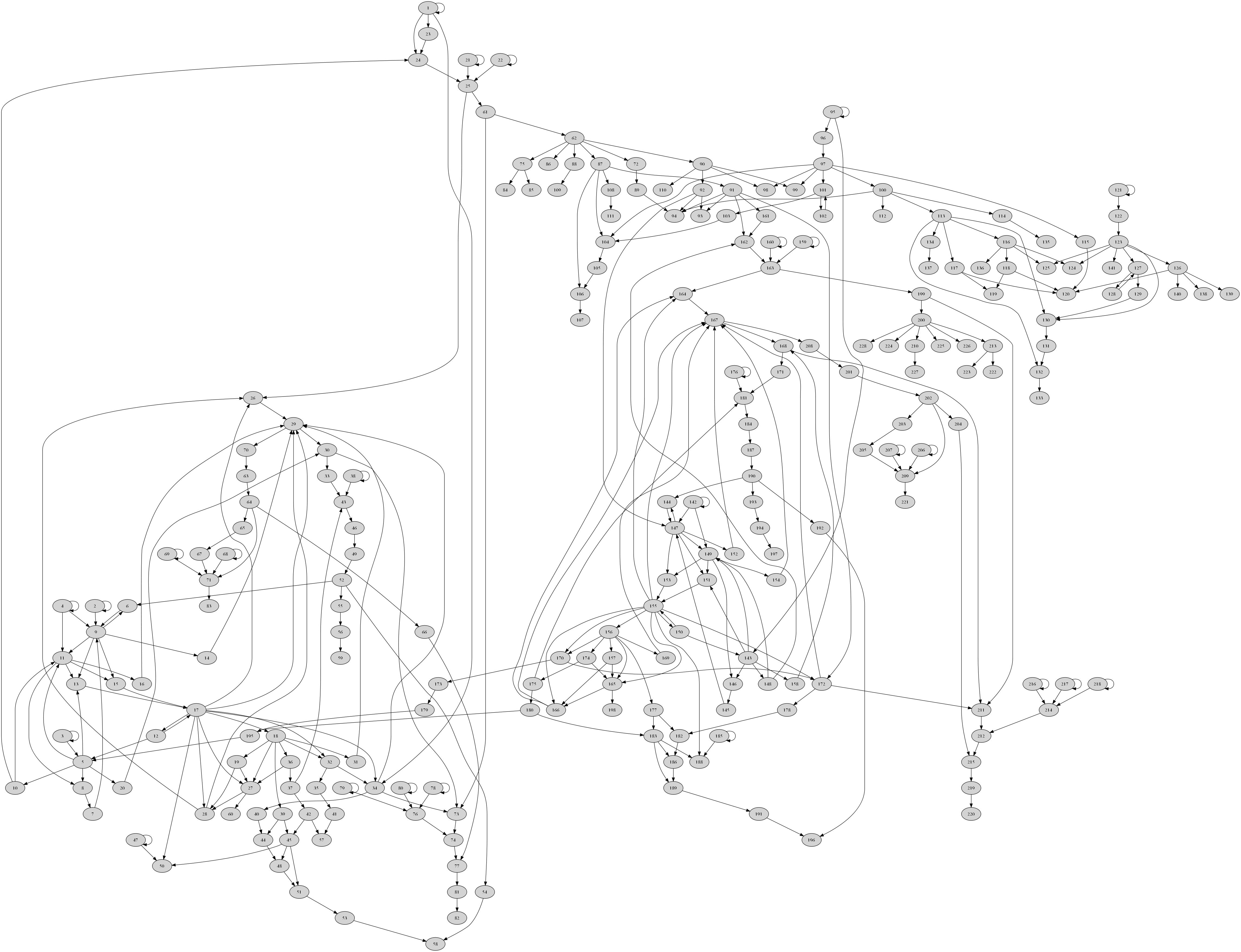}
\caption{{\footnotesize The dependency graph of a Boolean network with $228$ nodes and $25165824$ fixed points. Zoom in for detail}} \label{Fig. 1}
\end{center}
\end{figure}  
\par
\medskip
{\bf Example 3.5.} This is a simulated Boolean network with $450$ nodes and $2050$ interactions. The dependency graph (Fig. 2) is generated from the polynomial functions (see supplement MAPLE printout). With the threshold level at $<21$, the whole system was divided into $48$ subsystems, the dividing time was $19.16$ seconds, in which $12.84$ seconds were used in reading the inputed network. The total time for solving these $48$ subsystems was $8.3435$ minutes, and the total time for combining these solutions was $1.7045$ minutes.  $6$ fixed points were detected.
\par
\begin{figure}[h]
\begin{center}
\includegraphics[width=3in, height=3in]{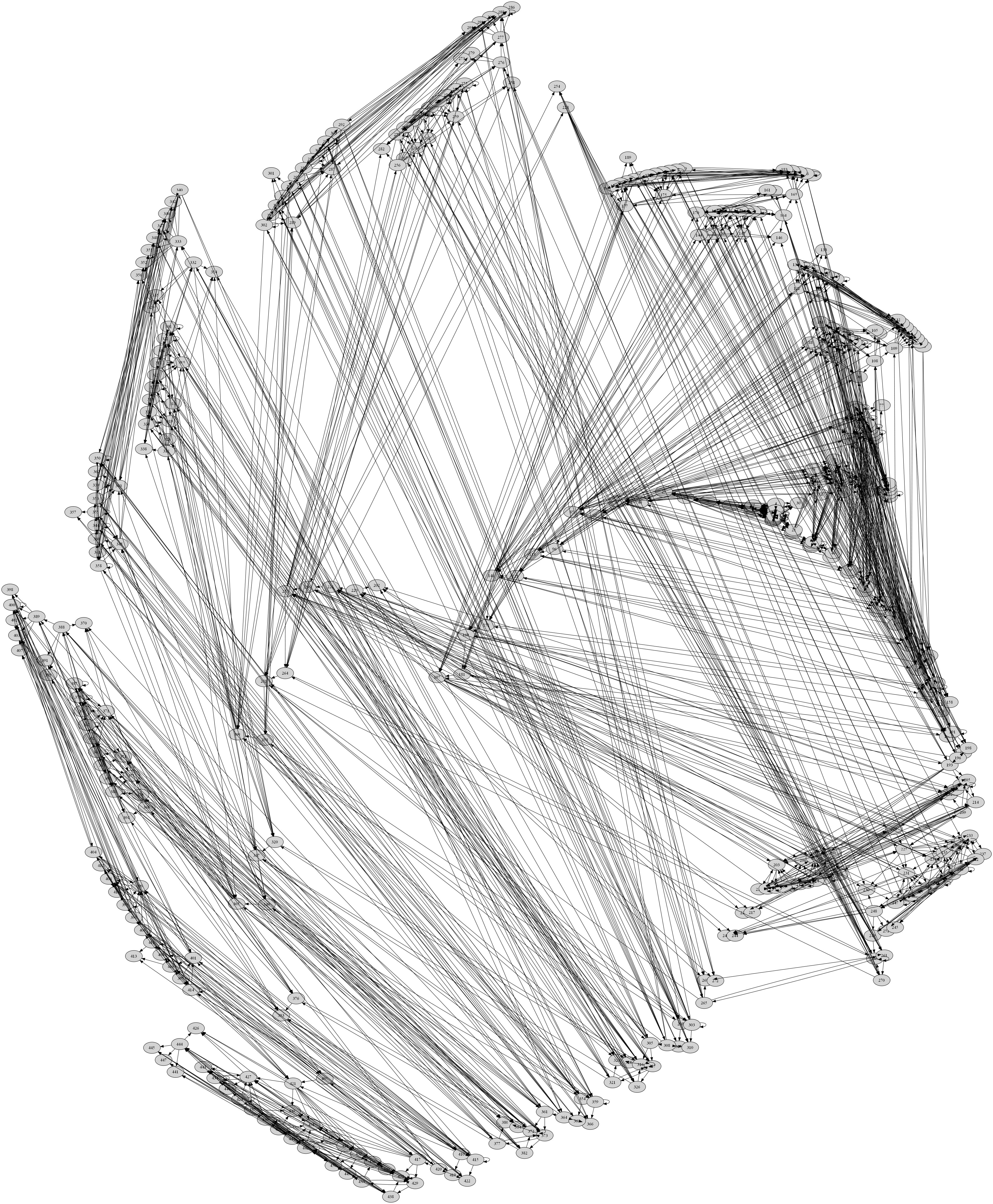}
\caption{{\footnotesize The dependency graph of a Boolean network with $450$ nodes and $2050$ directed edges. Zoom in for detail.}} \label{Fig. 2}
\end{center}
\end{figure}  
\par
\section{Concluding Remarks}
\par
\par
We have developed a new approach to solve systems of Boolean equations. With the computation of the fixed points of complex biological Boolean networks in mind, we developed our approach based on the characteristic of these networks, though it also applies to Boolean networks broadly. Our algorithm is self-contained, not an application of other algorithms, and thus it applies to Boolean networks beyond those have been considered before.  The approach is especially adaptable to large networks assembled from smaller components \cite{Purn09}, since these networks are naturally divisible. To demonstrate the effectiveness of our algorithm, we provided several examples of Boolean networks. The first two examples were included to show that exhaustive enumeration method can solve this problem for Boolean networks of sizes between $20$ and $30$ in less than a second with today's standard PCs, which provides the supporting evidence for our approach. The third example is the Boolean network published in \cite{SR07}. According to the authors, this network was the largest Boolean model of a cellular network known to them at the time of publication. Our algorithm used less than $30$ seconds to detect all fixed points of this Boolean network using MAPLE. The two simulated examples are substantially larger than the one in \cite{SR07}.  Thus we believe that our method will offer a useful tool for analyzing Boolean models, in particular, Boolean networks of biological systems.
\par


\section*{Appendix}
We provide the correspondences between gene names and the variables for the three Boolean networks cited from the literature. We refer the reader to the references for the original networks. The lengthy polynomial systems of Example 4 and 5 are provided in the MAPLE sheets.
\subsection*{Boolean Networks of \cite{Albe03}}
\par
We introduce the variables as follows:
\par
\begin{table}[h]
  \begin{center}
  {\renewcommand{\arraystretch}{1.4}\small {\footnotesize
  \begin{tabular}{|c|c|c|c|c|c|c|c|}
    \hline
    SLP & $wg$ & WG & $en$ & EN & $hh$& HH \\\hline
    $x_1$ & $x_2$ & $x_3$ & $x_4$ & $x_5$ & $x_6$ & $x_7$ \\\hline \hline
    $ptc$ & PTC & PH & SMO & $ci$ & CI & CIA\\ \hline
    $x_8$ & $x_9$ & $x_{10}$ & $x_{11}$ & $x_{12}$ & $x_{13}$ & $x_{14}$\\ \hline\hline
    CIR & WG$_{i-1}$ & WG$_{i+1}$ & HH$_{i-1}$ & HH$_{i+1}$ & $hh_{i-1}$ & $hh_{i+1}$ \\\hline
    $x_{15}$ & $x_{16}$ & $x_{17}$ & $x_{18}$ & $x_{19}$ & $x_{20}$ & $x_{21}$ \\\hline 
 \end{tabular}}
  }
  \label{tab: pairing1}
  \vskip 0.5cm
  \caption[Legend of variable names.]{\tiny Legend of variable names of the Boolean network in \cite{Albe03}.}
    \end{center}
\end{table} 
\par
Then the Boolean network is given by the following polynomial functions:
{\footnotesize \begin{align*}
f_1&=x_1,\quad
f_2 =(x_{15}+1)(x_1(x_2+x_{14})+x_2x_{14}),\quad
f_3 =x_2,\\
f_4&=x_1(x_{16}(x_{17}+1)+x_{17})+x_{16}(x_{17}+1)+x_{17},\\
f_5&=x_4,\quad
f_6 =x_5(x_{15}+1),\quad
f_7 =x_6,\\
f_8&=(x_4+1)x_{13}((x_{11}+1)(x_{20}(x_{21}+1)+x_{21})+x_{11}),\\
f_9&=(x_8+1)x_9(x_{18}+1)(x_{19}+1)+x_8,\\
f_{10}&=((x_8+1)x_9(x_{18}+1)(x_{19}+1)+x_8)(x_{20}(x_{21}+1)+x_{21}),\\
f_{11}&=f_{9}+f_{10}+1,\quad
f_{12} =x_5+1,\quad
f_{13} =x_{12},\\
f_{14}&=x_{13}((x_{11}+1)(x_{21}+1)(x_{20}+1)+1),\\
f_{15}&=f_{14}+x_{13},\quad
f_{i} =x_i\ \text{for }16\le i\le 21.
\end{align*}}
\par
\subsection*{The Boolean Network of \cite{ZhangR08}}
\par
This is the Boolean network given by the diagram of Fig. 2B in \cite{ZhangR08}. We introduce the variables as follows:
\begin{table}[h]
  \begin{center}
  {\renewcommand{\arraystretch}{1.4}\small{\tiny 
  \begin{tabular}{|c|c|c|c|c|c|c|c|c|}
    \hline
   IL15	&  RAS	& ERK	&    JAK	& IL2RBT	& STAT3	& IFNGT	& FasL\\ \hline 	
   $x_1$ & $x_2$ & $x_3$ & $x_4$ & $x_5$	& $x_6$	& $x_7$	& $x_8$\\ \hline\hline

   PDGF  &  PDGFR &   PI3K	 &   IL2	&   BcIxL	  &       TPL2	  &  SPHK &	S1P	  \\ \hline	
   $x_9$ & $x_{10}$ & $x_{11}$ & $x_{12}$ & $x_{13}$ & $x_{14}$ & $x_{15}$ & $x_{16}$ \\ \hline\hline

   sFas   & Fas &  DISC  &  Caspase  &  Apoptosis	&\cellcolor{Mygrey} &\cellcolor{Mygrey} &\cellcolor{Mygrey}\\ \hline
   $x_{17}$ & $x_{18}$ & $x_{19}$ & $x_{20}$ & $x_{21}$ & \cellcolor{Mygrey} &\cellcolor{Mygrey} &\cellcolor{Mygrey}\\ \hline\hline
   LCK	 &  MEK	  &  GZMB	&  IL2RAT	&   FasT	&  RANTES	&  A20	&   FLIP\\ \hline	
  $x_{22}$ & $x_{23}$ & $x_{24}$ & $x_{25}$ & $x_{26}$ & $x_{27}$ & $x_{28}$ & $x_{29}$\\\hline 
 \end{tabular}}
  }
  \label{tab: pairing2}
  \vskip 0.5cm
  \caption[Legend of variable names.]{\footnotesize Legend of variable names of the Boolean network in \cite{ZhangR08}.}
    \end{center}
\end{table} 
\par
Then the Boolean network is given by 
{\footnotesize \begin{align*}
f_1 &= f_2 = f_4 = f_5 = f_{22} = x_1,\quad
f_3  = f_{23} = x_2,\\
f_6 &= f_{24} = x_4,\quad
f_7  = x_5 + x_6 + x_5x_6,\\
f_8 &= x_6(x_3 + x_5 + x_3x_5) + x_{14} + x_6(x_3 + x_5 + x_3x_5)x_{14},\\
f_9 &= f_{10} = x_9,\quad
f_{11}  = x_{10},\\
f_{12} &= f_{13} = x_4 + x_{11} + x_4x_{11} +1,\quad
f_{14}  = f_{29} = x_{11},\\
f_{15} &= x_{11} + x_{16} + x_{11}x_{16},\quad
f_{16}  = f_{17} = x_{15},\\
f_{18} &= x_{17} + 1 + (x_1 + 1)(x_{11} + 1) + (x_{17} + 1)(x_1 + 1)(x_{11} + 1),\\
f_{19} &= x_{18},\quad
f_{20}  = (x_1 + 1)x_{19},\quad
f_{21}  = x_{20},\\
f_{25} &= x_{12},\quad
f_{26}  = f_{27}= f_{28} = x_{14}.
\end{align*}}
\par
\medskip
\newpage
\subsection*{The Boolean network of \cite{SR07}}
We introduce the variables as follows:
\begin{table}[h]
  \begin{center}
  {\renewcommand{\arraystretch}{1.4}\small{\tiny 
  \begin{tabular}{|c|c|c|c|c|c|c|c|c|c|}
    \hline
  
CD28   &     CD4 &	TCRIig &	CD45 &	TCRb &	SHP1 &	Csk &	PAG \\ \hline
$x1$ &	$x2$ &	$x3$ &	$x4$ &	$x5$ &	$x6$ &	$x7$ &	$x8$ \\ \hline\hline

Lckp1 &	Lckp2   & Fyn  &	CCbIp1  & TCRp &	RIK &	AbI &	cCbIp2  \\ \hline
$x9$ &	$x10$ & $x11$ &	$x12$ &	$x13$ &	$x14$ &	$x15$ &	$x16$  \\ \hline\hline

ZAP70 &	LAT &	Gads &	DGK  & SHIP-1 &	PTEN &	CbIb & 	PI3K  \\ \hline
$x17$ &	$x18$ &	$x19$ &	$x20$ & $x21$ &	$x22$ &	$x23$ & 	$x24$  \\ \hline\hline

PIP3 &	ItK &	Gab2 &	SLP76 &	PLCga &	DAG &  PLCgb & sh3bp2 \\ \hline
$x25$ &	$x26$ &	$x27$ &	$x28$ &	$x29$ &	$x30$ & $x31$ &	$x32$  \\ \hline\hline

 RasGRP1 & Vav1 & Vav3 &	Grb2 &	Sos  &	GAP5 &	HPK1 &	Rac1p1\\ \hline
$x33$ &	 $x34$ &	$x35$ &	$x36$ &	$x37$ &	$x38$ &	$x39$ &	$x40$\\ \hline\hline

Rac1p2 & Cdc42 &  Ra5 &	MLK3 &	MEKK1 & Raf &	Gadd45 &  MKK4  \\ \hline
$x41$ &	$x42$ &	$x43$ &	$x44$ &	$x45$ &	$x46$  &	 $x47$ &	  $x48$ \\ \hline\hline

MEK &  P38  &   JNK  &       ERK &   Jun &	Fos &	Rsk &	CREB  \\ \hline
$x49$ & $x50$  & $x51$ &	$x52$ &	$x53$ &	$x54$ &	$x55$ &	$x56$  \\ \hline

SRE &	AP1 &	CRE &	SHP2  & PDK1 &	PKB &	Ca &	CaM  \\ \hline
$x57$ &	$x58$ &	$x59$ &	$x60$  & $x61$ &	$x62$ &	$x63$ &	$x64$  \\ \hline\hline

CaMK4 & CaMK2 & CabIn1 & AKAP79 & CaIpr1 & IP3 & CaIcIn &	BAD \\ \hline
$x65$ &	$x66$ &	$x67$ &	$x68$ &	 $x69$ &	$x70$ & $x71$ &	$x72$  \\ \hline\hline

PKCth &	Ikkg &	GSK3 &	CARD11a & Ikkab &  CARD11  & BcI10 & MaIt1 \\ \hline
$x73$ &	$x74$ &	$x75$ &	$x76$ &	    $x77$ &  $x78$ &  $x79$ &  $x80$\\ \hline\hline

IkB &	NFkB &	NFAT &	bcat &	Cyc1 &	P21c &	p27k &	FKHR\\ \hline
$x81$ &	$x82$ &	$x83$ &	$x84$ &	$x85$ &	$x86$ &	$x87$ &	$x88$  \\ \hline\hline

 BcIXL &	p70S6k &\cellcolor{Mygrey} &\cellcolor{Mygrey} &\cellcolor{Mygrey} &\cellcolor{Mygrey} &\cellcolor{Mygrey} &\cellcolor{Mygrey} \\ \hline
$x89$ &	$x90$ &\cellcolor{Mygrey} &\cellcolor{Mygrey} &\cellcolor{Mygrey}&\cellcolor{Mygrey} &\cellcolor{Mygrey} &\cellcolor{Mygrey}\\ \hline
 \end{tabular}}
  }
  \label{tab: pairing3}
  \vskip 0.5cm
  \caption[Legend of variable names.]{\footnotesize Legend of variable names of the Boolean network in \cite{SR07}.}
    \end{center}
\end{table} 
\par\medskip
\newpage
Then the Boolean network is given by
{\footnotesize \begin{align*}
f_{1}  &= x_{1},\quad 
f_{2}  = x_{2},\quad 
f_{3}  = x_{3},\quad 
f_{4}  = x_{4},\quad \\
f_{5}  &= x_{3}*(x_{12}+1),\quad 
f_{6}  = x_{9}*(x_{52}+1),\quad 
f_{7}  = x_{8},\quad \\
f_{8}  &= x_{5}+1+x_{5}*x_{11},\quad 
f_{9}  = x_{2}*x_{4}*(x_{6}+1)*(x_{7}+1),\quad 
f_{10}  = x_{5},\quad \\
f_{11}  &= x_{5}*(x_{9}+x_{10}+x_{9}*x_{10})+x_{4}*x_{9}*(x_{5}+1),\quad \\
f_{12} & = x_{17},\quad 
f_{13}  = x_{5}*(x_{9}+x_{11}+x_{9}*x_{11}),\quad \\
f_{14} & = x_{9},\quad 
f_{15}  = x_{9}+x_{11}+x_{9}*x_{11},\quad 
f_{16}  = x_{11},\quad \\
f_{17}  &= (x_{12}+1)*x_{13}*x_{15},\quad 
f_{18}  = x_{17},\quad 
f_{19}  = x_{18},\quad \\
f_{20} & = x_{5},\quad 
f_{21}  = x_{21},\quad 
f_{22}  = x_{22},\quad 
f_{23}  = x_{1}+1,\quad \\
f_{24}  &= (x_{1}+x_{10}+x_{1}*x_{10})*(x_{23}+1),\quad 
f_{25}  = x_{24}*(x_{21}+1)*(x_{22}+1),\quad \\ 
f_{26}  &= x_{17}*x_{25}*x_{28},\quad 
f_{27}  = x_{17}*x_{18}*(x_{19}+x_{36}+x_{19}*x_{36}),\quad \\
f_{28}  &= x_{17}*x_{19}*(x_{27}+1),\quad \\
f_{29}  &= x_{17}*x_{28}*x_{31}*x_{34}*(x_{26}+x_{14}*(x_{16}+1)+x_{26}*x_{14}*(x_{16}+1)),\quad \\
f_{30}  &= x_{29}*(x_{20}+1),\quad 
f_{31}  = x_{18},\quad 
f_{32}  = x_{17}*x_{18},\quad 
f_{33}  = x_{30},\quad \\
f_{34}  & = x_{1}+x_{17}*x_{32}+x_{1}*x_{17}*x_{32},\quad \\
f_{35}  &= x_{32},\quad 
f_{36}  = x_{18},\quad 
f_{37}  = x_{36},\quad 
f_{38}  = x_{38},\quad \\
f_{39}  &= x_{18},\quad 
f_{40}  = x_{34},\quad 
f_{41}  = x_{35},\quad 
f_{42}  = x_{37},\quad \\
f_{43}  &= x_{33}*x_{37}*(x_{38}+1),\quad 
f_{44}  = x_{39}+x_{40}+x_{40}*x_{39},\quad \\
f_{45}  &= x_{39}+x_{42}+x_{39}*x_{42},\quad 
f_{46}  = x_{43},\quad 
f_{47}  = x_{47},\quad \\
f_{48} & = x_{44}+x_{45}+x_{44}*x_{45},\quad 
f_{49}  = x_{46},\quad \\
f_{50} & = x_{45}+x_{17}*(x_{47}+1)+x_{45}*x_{17}*(x_{47}+1),\quad \\
f_{51}  &= x_{45}+x_{48}+x_{45}*x_{48},\quad 
f_{52}  = x_{49},\quad \\
f_{53}  &= x_{51},\quad 
f_{54}  = x_{52},\quad 
f_{55}  = x_{52},\quad 
f_{56}  = x_{55},\quad \\
f_{57}  & = x_{41}+x_{42}+x_{41}*x_{42},\quad 
f_{58}  = x_{53}*x_{54},\quad \\
f_{59}  &= x_{56},\quad 
f_{60}  = x_{27},\quad 
f_{61}  = x_{25},\quad 
f_{62}  = x_{61},\quad \\
f_{63}  &= x_{70},\quad 
f_{64}  = x_{63},\quad 
f_{65}  = x_{64},\quad 
f_{66}  = x_{64},\quad \\
f_{67}  &= x_{65}+1,\quad 
f_{68}  = x_{68},\quad 
f_{69}  = x_{69},\quad 
f_{70}  = x_{29},\quad \\
f_{71}  &= x_{64}*(x_{67}+1)*(x_{68}+1)*(x_{69}+1),\quad \\
f_{72}  &= x_{62}+1,\quad 
f_{73}  = x_{30}*x_{34}*x_{61},\quad 
f_{74}  = x_{73}*x_{76},\quad \\
f_{75}  & = x_{62}+1,\quad 
f_{76}  = x_{78}*x_{79}*x_{80},\quad 
f_{77}  = x_{66}*x_{74},\quad \\
f_{78}  & = x_{78},\quad 
f_{79}  = x_{79},\quad 
f_{80}  = x_{80},\quad 
f_{81}  = x_{77}+1,\quad \\
f_{82}  & = x_{81}+1,\quad 
f_{83}  = x_{71},\quad 
f_{84}  = x_{75}+1,\quad 
f_{85}  = x_{75}+1,\quad \\
f_{86}  & = x_{62}+1,\quad 
f_{87}  = x_{62}+1,\quad 
f_{88}  = x_{62}+1,\quad \\
f_{89}  &= x_{72}+1,\quad 
f_{90}  = x_{62}.
\end{align*}}
\end{document}